\pdfoutput=1
\documentclass[11pt]{myclass}

\usepackage{epsfig}
\usepackage{amsmath}
\usepackage{amssymb}
\usepackage{euscript}

\newcommand{\eps}{\varepsilon}

\renewcommand{\b}[1]{\ensuremath{\mathbb{#1}}}

\newcommand{\EuD}{\ensuremath{\EuScript{D}}}
\newcommand{\EuA}{\ensuremath{\EuScript{A}}}
\newcommand{\EuI}{\ensuremath{\EuScript{I}}}
\newcommand{\EuR}{\ensuremath{\EuScript{R}}}
\newcommand{\EuP}{\ensuremath{\EuScript{P}}}
\newcommand{\EuZ}{\ensuremath{\EuScript{Z}}}
\newcommand{\EuM}{\ensuremath{\EuScript{M}}}
\newcommand{\EuC}{\ensuremath{\EuScript{C}}}
\newcommand{\EuS}{\ensuremath{\EuScript{S}}}
\newcommand{\EuQ}{\ensuremath{\EuScript{Q}}}

\usepackage[section]{algorithm}
\usepackage{algorithmic}
\usepackage{graphics}
\usepackage{graphicx}

\title{An Efficient Algorithm for 2D Euclidean 2-Center with Outliers\thanks{This work is supported by NSF under grants CNS-05-40347, CFF-06-35000, and DEB-04-25465, by ARO grants W911NF-04-1-0278 and W911NF-07-1-0376, by an NIH grant 1P50-GM-08183-01, by a DOE grant OEGP200A070505, and by a grant from the U.S. Israel Binational Science Foundation.}}
\author{Pankaj K. Agarwal
\thanks{Department of Computer Science, Duke University, Durham, NC 27708: \texttt{pankaj@cs.duke.edu}} 
\and Jeff M. Phillips
\thanks{Department of Computer Science, Duke University, Durham, NC 27708: \texttt{jeffp@cs.duke.edu}}
}

\begin{document}
\maketitle

\begin{abstract}
For a set $P$ of $n$ points in $\b{R}^2$, the Euclidean $2$-center problem computes a pair of congruent disks of the minimal radius that cover $P$.  We extend this to the $(2,k)$-center problem where we compute the minimal radius pair of congruent disks to cover $n-k$ points of $P$.  
We present a randomized algorithm with $O(n k^7 \log^3 n)$ expected running time for the $(2,k)$-center problem.  
We also study the $(p,k)$-center problem in $\b{R}^2$ under the $\ell_\infty$-metric for $p=\{4,5\}$. 
We propose an $k^{O(1)} n\log n$ algorithm for computing a $\ell_\infty$ $(4,k)$-center and an $k^{O(1)} n \log^5 n$ algorithm for computing a $\ell_\infty$ $(5,k)$-center. 
\end{abstract}

\section{Introduction}
Let $P$ be a set of $n$ points in $\b{R}^2$.  For a pair of integers $0\leq k \leq n$ and $p \geq 1$, a family of $p$ congruent disks is called a \emph{$(p,k)$-center} if the disks cover at least $n-k$ points of $P$; $(p,0)$-center is the standard $p$-center.  The Euclidean $(p,k)$-center problems asks for computing a $(p,k)$-center of $P$ of the smallest radius.  In this paper we study the $(2,k)$-center problem.
We also study the $(p,k)$-center problem under the $\ell_\infty$-metric for small values of $p$ and $k$.  Here we wish to cover all but $k$ points of $P$ by $p$ congruent axis-aligned squares of the smallest side length.  Our goal is to develop algorithms whose running time is $n (k \log n)^{O(1)}$.

\paragraph{\textbf{\emph{Related work.}}}
There has been extensive work on the $p$-center problem in algorithms and operations research communities~\cite{AS98,DH02,Hoch95,CKMN01}.  
If $p$ is part of the input, the problem is NP-hard~\cite{MS84} even for the Euclidean case in $\b{R}^2$.  The Euclidean $1$-center problem is known to be LP-type~\cite{MSW96}, and therefore can be solved in linear time for any fixed dimension.  
The Euclidean $2$-center problem is not LP-type.  Agarwal and Sharir~\cite{AS94} proposed an $O(n^2 \log^3 n)$ time algorithm for the 2-center problem.  The running time was improved to $O(n \log^{O(1)} n)$ by Sharir~\cite{Sha97}.  The exponent of the $\log n$ factor was subsequently improved in \cite{Epp97,Chan99}.  The best known deterministic algorithm takes $O(n \log^2 n \log^2 \log n)$ time in the worst case, and the best known randomized algorithm takes $O(n \log^2 n)$ expected time.


There is little work on the $(p,k)$-center problem.
Using a framework described by Matou\v{s}ek~\cite{Mat95}, LP-type problems, with $k$ violations and basis size $3$, can be solved in $O(n \log k + k^3 n^\eps)$ time, for any $\eps > 0$.  This is improved by Chan~\cite{Cha05} to $O(n \beta(n) \log n + k^2 n^{\eps})$ expected time, where $\beta(\cdot)$ is a slow-growing inverse-Ackermann-like function and $\eps>0$.  
The $(1,k)$-center problem is LP-type with basis size $3$, so these bounds apply.  
Matou\v{s}ek~\cite{Mat95} also gives a more general results for LP-type problems with $k$ violations and with basis size $c$ that runs in $O(n k^c)$ time, if it is a feasible case where a solution with no violations exists.  In the infeasible case, no solution exists without violations and the algorithm runs in $O(n k^{c+1})$ time.  In fact, he shows in the feasible (resp. infeasible) case that there are $O(k^c)$ (resp. $O(k^{c+1})$) basis with at most $k$ violations and his algorithm visits all of them by a path of length $O(k^c)$ (resp. $O(k^{c+1})$) where consecutive basis in the path differ by inserting or deleting one constraint.  

The $p$-center problem under $\ell_\infty$-metric is dramatically simpler.  
Sharir and Welzl~\cite{SW96} show how to compute the $\ell_\infty$ $p$-center in near-linear time for $p \leq 5$.  In fact, they show that the rectilinear $2$- and $3$-center problems are LP-type problems and can be solved in $O(n)$ time.  Also, they show the 1-dimensional version of the problem is an LP-type problem for any $p$, with combinatorial dimension $O(p)$.  Thus applying Matou\v{s}ek's framework \cite{Mat95}, the $\ell_\infty$ $(p,k)$-center in $\b{R}^2$ for $p \leq 3$, can be found in $O(k^{O(1)} n)$ time and in $O(k^{O(p)} n)$, for any $p$, if the points lie in $\b{R}^1$.  


\paragraph{\textbf{\emph{Our results.}}}
Our main result is a randomized algorithm for the Euclidean $(2,k)$-center problem in $\b{R}^2$ whose expected running time is $O(nk^7 \log^3 n)$.  
We follow the general framework of Sharir and subsequent improvements by Eppstein.  
We first prove, in Section 2, a few  structural properties of levels in an arrangement of unit disks, which are of independent interest.  

As in \cite{Sha97,Epp97}, our solution breaks the $(2,k)$-center problem into two cases depending on the distance between the centers of the optimal disks; 
(i) the centers are further apart than the optimal radius, and 
(ii) they are closer than their radius.  The first subproblem, which we refer to as the \emph{well-separated case} and describe in Section 3, takes $O(k^6 n \log^3 n)$ time in the worst case and uses parametric search~\cite{Meg83}.  The second subproblem, which we refer to as the \emph{nearly concentric case} and describe in Section 4, takes $O(k^7 n \log^3 n)$ expected time.  Thus we solve the $(2,k)$-center problem in $O(k^7 n \log^3 n)$ expected time.  
We can solve the nearly concentric case and hence the $(2,k)$-center problem in $O(k^7 n^{1+\delta})$ deterministic time, for any $\delta>0$.  
We present near-linear algorithms for the $\ell_\infty$ $(p,k)$-center in $\b{R}^2$ for $p={4,5}$.  The $\ell_\infty$ $(4,k)$-center problem takes $O(k^{O(1)} n \log n)$ time, and the $\ell_\infty$ $(5,k)$-center problem takes $O(k^{O(1)} n \log^5 n)$ time.  
We have not made any attempt to minimize the exponent of $k$.  We believe that it can be improved by a more careful analysis.

\section{Arrangement of Unit Disks}
\label{sec:arrangementD}
Let $\EuD = \{D_1, \ldots, D_n\}$ be a set of $n$ unit disks in $\b{R}^2$.  Let $\EuA(\EuD)$ be the arrangement of $\EuD$.\footnote{The \emph{arrangement} of $\EuD$ is the planar decomposition induced by $\EuD$; its vertices are the intersection points of boundaries of two disks, its edges are the maximal portions of disk boundaries that do not contain a vertex, and its faces are the maximal connected regions of the plane that do not intersect the boundary of any disk.}  $\EuA(\EuD)$ consists of $O(n^2)$ vertices, edges, and faces.  For a subset $\EuR \subseteq \EuD$, let $\EuI(\EuR) = \bigcap_{D \in \EuR} D$ denote the intersection of disks in $\EuR$.  Each disk in $\EuR$ contributes at most one edge in $\EuI(\EuR)$.  We refer to $\EuI(\EuR)$ as a \emph{unit-disk polygon} and a connected portion of $\partial \EuI(\EuR)$ as a \emph{unit-disk curve}.  
We introduce the notion of a level in $\EuA(\EuD)$, prove a few structural properties of levels, and describe a procedure that will be useful for our overall algorithm.  

\paragraph{\textbf{\emph{Levels and their structural properties.}}}
For a point $x \in \b{R}^2$, the \emph{level} of $x$ with respect to $\EuD$, denoted by $\lambda(x,\EuD)$, is the number of disks in $\EuD$ that \emph{do not} contain $x$.  (Our definition of level is different from the more common definition in which it is defined as the number of disks whose interiors contain $x$.)  All points lying on an edge or face $\phi$ of $\EuA(\EuD)$ have the same level, which we denote by $\lambda(\phi)$.  For $k \leq n$, let $\EuA_k(\EuD)$ (resp. $\EuA_{\leq k}(\EuD)$) denote the set of points in $\b{R}^2$ whose level is $k$ (resp. at most $k$); see Figure \ref{fig:arrD}.  By definition, $\EuA_0(\EuD) = \EuA_{\leq 0}(\EuD) = \EuI(\EuD)$.  

The boundary of $\EuA_{\leq k}(\EuD)$ is composed of the edges of $\EuA(\EuD)$.  
Let $v \in \partial D_1 \cap \partial D_2$, for $D_1, D_2 \in \EuD$, be a vertex of $\partial \EuA_{\leq k}(\EuD)$.  We call $v$ \emph{convex} (resp. \emph{concave}) if $\EuA_{\leq k}(\EuD)$ lies in $D_1 \cap D_2$ (resp. $D_1 \cup D_2$) in a sufficiently small neighborhood of $v$; see Figure \ref{fig:arrD}(a).  $\partial \EuA_{\leq 0}(\EuD)$ is composed of convex vertices.
We define the complexity of $\EuA_{\leq k}(\EuD)$ to be the number of edges of $\EuA(\EuD)$ whose levels are at most $k$.  Since the complexity of $\EuA_{\leq 0}(\EuD)$ is $n$, the following lemma follows from the result by Clarkson and Shor~\cite{CS89} (see also Sharir~\cite{Sha91} and Chan~\cite{Cha07}).  

\begin{figure}
  \centering
  \includegraphics{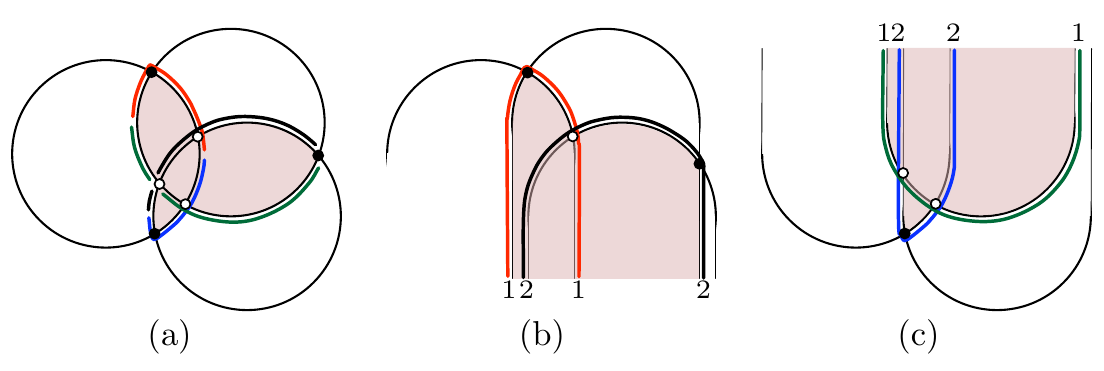}
\caption{\label{fig:arrD} 
\small (a) $\EuA(\EuD)$; shaded region is $\EuA_{\leq 1}(\EuD)$; filled (resp. hollow) vertices are convex (resp. concave) vertices of $\EuA_{\leq 1}(\EuD)$; covering of $\EuA_{\leq 1}(\EuD)$ edges by six unit-disk curves.  
(b) $\EuA(\Gamma^+)$; shaded region is $\EuA_{\leq 1}(\Gamma^+)$; and the covering of $\EuA_{\leq 1}(\Gamma^+)$ edges by two concave chains.
(c) $\EuA(\Gamma^-)$; shaded region is $\EuA_{\leq 1}(\Gamma^-)$; and the covering of $\EuA_{\leq 1}(\Gamma^-)$ edges by two convex chains.}
\end{figure}

\begin{lemma} \emph{\cite{CS89}} For $k \geq 0$, the complexity of $\EuA_{\leq k}(\EuD)$ is $O(nk)$.  
\label{lem:CS-nk}
\end{lemma}


\begin{remark}
The argument by Clarkson and Shor can also be used to prove that $\EuA_{\leq k}(\EuD)$ has $O(k^2)$ connected components and that it has $O(k^2)$ local minima in $(+y)$-direction.  See also \cite{Cla93,Mat95}.  These bounds are tight in the worst case; see Figure \ref{fig:lb-k2}.
\end{remark}

\vspace{.1in}
It is well known that the edges in the ${\le}k$-level of a line arrangement can be covered by $k+1$ concave chains~\cite{Gus79}, as used in \cite{Dey98,Cha05}.  We prove a similar result for $\EuA_{\leq k}(\EuD)$; it can be covered by $O(k)$ unit-disk curves.  

For a disk $D_i$, let $\gamma_i^+$ (resp. $\gamma_i^-$) denote the set of points that lie in or below (resp. above) $D_i$; $\partial \gamma_i^+$ consists of the upper semicircle of $\partial D_i$ plus two vertical downward rays emanating from the left and right endpoints of the semicircle --- we refer to these rays as left and right rays.  The curve $\partial \gamma_i^-$ has a similar structure.  See Figures \ref{fig:arrD}(b) and (c).  Set $\Gamma^+ = \{\gamma_i^+ \mid 1 \leq i \leq n\}$ and $\Gamma^- = \{\gamma_i^- \mid 1 \leq i \leq n\}$.  Assuming that the $x$-coordinates of the centers of all disks in $\EuD$ are distinct, each pair of curves $\partial \gamma_i^+, \partial \gamma_j^+$ intersect in at most one point.  (If we assume that the left and right rays are not vertical but have very large positive and negative slopes, respectively, then each pair of boundary curves intersects in exactly one point.)  We define the level of a point with respect to $\Gamma^+$, $\Gamma^-$, or $\Gamma^+ \cup \Gamma^-$ in the same way as with respect to $\EuD$.  A point lies in a disk $D_i$ if and only if it lies in both $\gamma_i^+$ and $\gamma_i^-$, so we obtain the following inequalities:
\begin{equation}
\max \{\lambda(x, \Gamma^+), \lambda(x,\Gamma^-) \} \leq \lambda(x,\EuD).
\label{eq:gam-leq-D}
\end{equation}
\begin{equation}
\lambda(x,\EuD) \leq \lambda(x, \Gamma^+ \cup \Gamma^-) \leq 2 \lambda(x,\EuD).
\label{eq:cup-D}
\end{equation}

We cover the edges of $\EuA_{\leq k}(\Gamma^+)$ by $k+1$ concave chains as follows.  The level of the $(k+1)$st rightmost left ray is at most $k$ at $y=-\infty$.  Let $\rho_i$ be such a ray, belonging to $\gamma_i^+$.  We trace $\partial \gamma_i^+$, beginning from the point at $y=-\infty$ on $\rho_i$, as long as $\partial \gamma_i^+$ remains in $\EuA_{\leq k}(\Gamma^+)$.  We stop when we have reached a vertex $v \in \EuA_{\leq k}(\Gamma^+)$ at which it leaves $\EuA_{\leq k}(\Gamma^+)$; $v$ is a convex vertex on $\EuA_{\leq k}(\Gamma^+)$.  Suppose $v = \partial \gamma_i^+ \cap \partial \gamma_j^+$.  Then $\partial \EuA_{\leq k}(\Gamma^+)$ follows $\partial \gamma_j^+$ immediately to the right of $v$, so we switch to $\partial \gamma_j^+$ and repeat the same process.  It can be checked that we finally reach $y=-\infty$ on a right ray.  Since we always switch the curve on a convex vertex, the chain $\Lambda_i^+$ we trace is a concave chain composed of a left ray, followed by a unit-disk curve $\xi_i^+$, and then followed by a right ray.  Let $\Lambda_0^+, \Lambda_1^+, \ldots, \Lambda_k^+$ be the $k+1$ chains traversed by this procedure.  These chains cover all edges of $\EuA_{\leq k}(\Gamma^+)$, and each edge lies exactly on one chain.  Similarly we cover the edges of $\EuA_{\leq k}(\Gamma^-)$ by $k+1$ convex curves $\Lambda_0^-, \Lambda_1^-, \ldots, \Lambda_k^-$.  Let $\Xi = \{\xi_0^+, \ldots, \xi_k^+, \xi_0^-, \ldots, \xi_k^-\}$ be the family of unit-disk curves induced by these convex and concave chains.  By (\ref{eq:gam-leq-D}), $\Xi$ covers all edges of $\EuA_{\leq k}(\EuD)$.  
Since a unit circle intersects a unit-disk curve in at most two points, we conclude the following.

\begin{lemma}
The edges of $\EuA_{\leq k}(\EuD)$ can be covered by at most $2k+2$ unit-disk curves, and a unit circle intersects $O(k)$ edges of $\EuA_{\leq k}(\EuD)$.  
\label{lem:k-udc}
\end{lemma}

The curves in $\Xi$ may contain edges of $\EuA(\EuD)$ whose levels are greater that $k$.  If we wish to find a family of unit-disk curves whose union is the set of edges in $\EuA_{\leq k}(\EuD)$, we proceed as follows.  
We add the $x$-extremal points of each disk as vertices of $\EuA(\EuD)$, so each edge is now $x$-monotone and lies in a lower or an upper semicircle.  By (\ref{eq:gam-leq-D}), only $O(k)$ such vertices lie in $\EuA_{\leq k}(\EuD)$.  
We call a vertex of $\EuA_{\leq k}(\EuD)$ \emph{extremal} if it is an $x$-extremal point on a disk or an intersection point of a lower and an upper semicircle.  An extremal vertex of the latter type is an intersection point of $\xi_i^+, \xi_i^- \in \Xi$.  Since each such pair intersects in at most two points, there are $O(k^2)$ extremal vertices.  
For each extremal vertex $v$ we do the following.  If there is an edge $e$ of $\EuA_{\leq k}(\EuD)$ lying to the right of $v$, we follow the arc containing $e$ until we reach an extremal vertex or we leave $\EuA_{\leq k}(\EuD)$.  In the former case we stop.  In the latter 
case we are at a convex vertex $v^\prime$ of $\partial \EuA_{\leq k}(\EuD)$, and we switch to the other arc incident on $v^\prime$ and continue.  These curves have been drawn in Figure \ref{fig:arrD}(a).  This procedure returns an $x$-monotone unit-disk curve that lies in $\EuA_{\leq k}(\EuD)$.  It can be shown that this procedure covers all edges of $\EuA_{\leq k}(\EuD)$.  
If $\EuA_{\leq k}(\EuD)$ is represented as a planar graph, we can compute these curves in time proportional to the number of edges in $\EuA_{\leq k}(\EuD)$.  
We thus obtain the following:

\begin{lemma}
Let $\EuD$ be a set of $n$ unit disks in $\b{R}^2$.  Given $\EuA_{\leq k}(\EuD)$, we can compute, in time $O(nk)$, a family of $O(k^2)$ $x$-monotone unit-disk curves whose union is the set of edges of $\EuA_{\leq k}(\EuD)$.  
\label{lem:udcs-k2}
\end{lemma}

\begin{figure}
  \centering 
  \includegraphics[width=.3\linewidth]{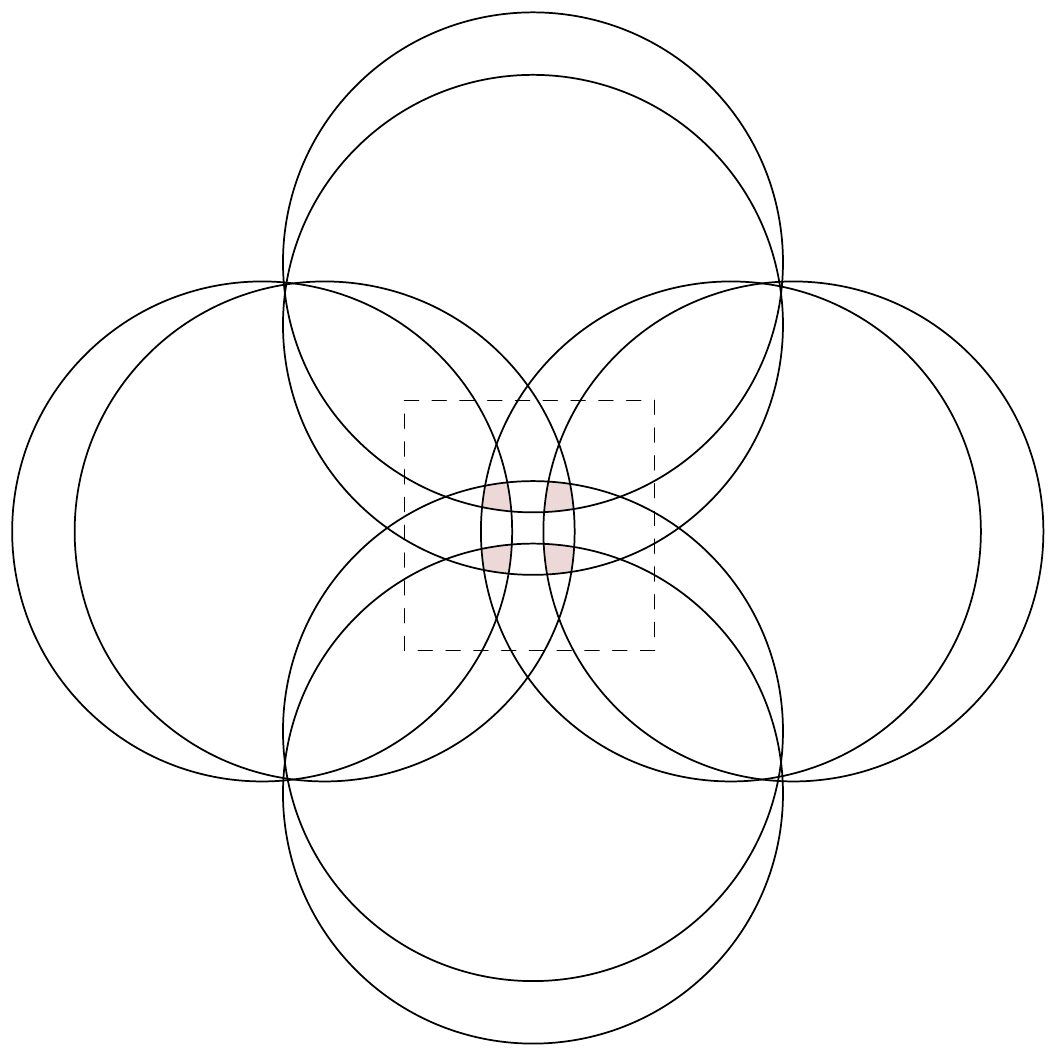}
  \hspace{.5in}
  \includegraphics[width=.3\linewidth]{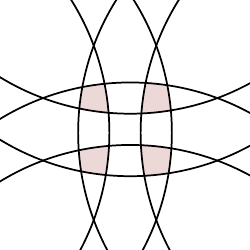}
\caption{\label{fig:lb-k2} 
Lower bound.  $\EuA_{\leq 2}(\EuD)$ (shaded region) has $4$ connected components.  The right image is zoomed in of the center of the left image.}
\end{figure}

\begin{remark}
Since $\EuA_{\leq k}(\EuD)$ can consist of $\Omega(k^2)$ connected components, the $O(k^2)$ bound is tight in the worst case; see Figure \ref{fig:lb-k2}.
\end{remark}

\paragraph{\textbf{\emph{Dynamic Data Structures.}}} 
We need a dynamic data structure for storing a set $\EuD$ of unit disks that supports the following two operations:
\begin{itemize}
\item{(O1)} Insert a disk into $\EuD$ or delete a disk from $\EuD$;
\item{(O2)} For a given $k$, determine whether $\EuA_{\leq k}(\EuD) \neq \emptyset$.
\end{itemize}

Hershberger and Suri~\cite{HS91}, describe how to maintain $\EuI(\EuD)$ under insertion/deletion in $O(\log n)$ time per update and how to find the point in $\EuI(\EuD)$ with the smallest $y$-coordinate in $O(\log n)$ time.
We use this in conjunction with Matou\v{s}ek's algorithm~\cite{Mat95} for visiting all basis of an LP-type problem with at most $k$ violations.   Specifically we examine the LP-type problem of finding the smallest $y$-coordinate of $\EuI(\EuD)$ with $k$ violations, which has a basis size of $2$ and may be infeasible.  Thus the path to visit all basis is of length $O(k^3)$ and using Hershberger and Suri's data structure we traverse each step of the path in $O(\log n)$ time by inserting or deleting a constraint and finding the discs defining the minimal $y$-coordinate.  

\begin{lemma}
There exists a dynamic data structure for storing a set of $n$ unit disks so that \emph{(O1)} can be performed in $O(\log n)$ time, and \emph{(O2)} takes $O(k^3 \log n)$ time.  
\label{lem:dynD}
\end{lemma}

Agarwal and Matou\v{s}ek~\cite{AM95} provide a data structure that can maintain the value of the radius of the smallest enclosing disk under insertions and deletions in $O(n^{\delta})$ time per update, for any $\delta>0$.  
We combine this with Matou\v{s}ek's algorithm for LP-type problems, specifically for the $(1,k)$-center problem.  Similar to the above data structure, the algorithm determines a path of length $O(k^3)$ to traverse all basis with at most $k$ violations, and each is traversed in $O(n^{\delta})$ time by handling an insertion or deletion using Agarwal and Matou\v{s}ek's data structure.  

\begin{lemma}
The exists a dynamic data structure for a set of $n$ points such that under insertion/deletion of a point, it can return the answer to the $(1,k)$-center problem in $O(k^3 n^\delta)$, for any $\delta>0$.  
\label{lem:dynC}
\end{lemma}

\section{Well-Separated Disks}
In this section we describe an algorithm for the case in which the two disks $D_1$, $D_2$ of the optimal solution are well separated.  That is, let $c_1$ and $c_2$ be the centers of $D_1$ and $D_2$, and let $r^*$ be their radius.  Then $||c_1 c_2|| \geq r^*$; see Figure \ref{fig:well-sep}.  Without loss of generality, let us assume that $c_1$ lies to the left of $c_2$.
Let $D_i^-$ be the semidisk lying to the left of the line passing through $c_1$ in direction normal to $c_1 c_2$.  
A line $\ell$ is called a separator line if $D_1 \cap D_2 = \emptyset$ and $\ell$ separates $D_1^-$ from $D_2$, or $D_1 \cap D_2 \neq \emptyset$ and $\ell$ separates $D_1^-$ from the intersection points $\partial D_1 \cap \partial D_2$.  
We first show that we can quickly compute a set of $O(k^2)$ lines that contains a separator line.  Next, we describe a decision algorithm, and then we describe the algorithm for computing $D_1$ and $D_2$ provided they are well separated.  

\paragraph{\textbf{\emph{Computing separator lines.}}}
We fix a sufficiently large constant $h$ and choose a set $U = \{ u_1, \ldots, u_h\} \subseteq \b{S}^1$ of directions, where $u_i = \left(\cos (2\pi i / h), \sin (2 \pi i / h)\right)$.  

\begin{figure}
  \centering
  \includegraphics{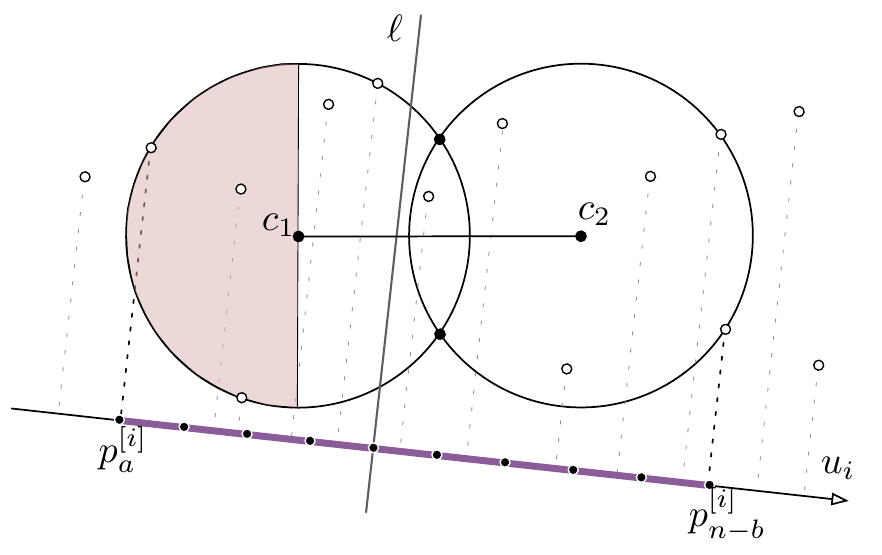}
\caption{\label{fig:well-sep} 
Let $\ell$ is a separator line for disks $D_1$ and $D_2$.
}
\end{figure}


For a point $p \in \b{R}^2$ and a direction $u_i$, let $p^{[i]}$ be the projection of $p$ in the direction normal to $u_i$.  Let $P^{[i]} = \langle p_1^{[i]}, \ldots, p_n^{[i]} \rangle$ be the sorted sequence of projections of points in the direction normal to $u_i$.  For each pair $a, b$ such that $a+b\leq k$, we choose the interval $\delta_{a,b}^{[i]} = [p_a^{[i]}, p_{n-b}^{[i]}]$ and we place $O(1)$ equidistant points in this interval.  See Figure \ref{fig:well-sep}(a).  Let $L_{a,b}^{[i]}$ be the set of (oriented) lines in the direction normal to $u_i$ and passing though these points.  Set 
$$
L = \bigcup_{\substack{1 \leq i \leq h \\ a+b \leq k}} L_{a,b}^{[i]}.
$$

The set $L$ can be computed in $O(k^2 n \log n)$ time.
We claim that $L$ contains at least one separator line.  
Let $u_i \in U$ be the direction closest to $\overrightarrow{c_1 c_2}$.  Suppose $p_a$ and $p_{n-b}$ are the first and the last points of $P$ in the direction $u_i$ that lie inside $D_1 \cup D_2$.  Since $|P \setminus (D_1 \cup D_2)| \leq k$, $a+b \leq k$.  
Let $q_1$ be the extreme point of $D_1^-$ in direction $u_i$ and let $q_2$ be the extreme point of $D_2 \setminus D_1$ in direction $-u_i$.
Since $u_i$ is within a small constant angle of $\overrightarrow{c_1 c_2}$ $$
\langle q_2 - q_1, u_i \rangle 
\geq 
\alpha \langle q_2 - q_1, \overrightarrow{c_1 c_2} \rangle 
= 
\frac{\alpha}{2} \langle c_2 - c_1, \overrightarrow{c_1 c_2} \rangle 
\geq 
\frac{\alpha}{6} \langle p_{n-b} - p_a, u_i \rangle,
$$ 
where $\alpha \leq 1$ is a constant depending on $h$.  Hence if at least $6/\alpha$ points are chosen in the interval $\delta_{a,b}^{[i]}$, then one of the lines in $L_{a,b}^{[i]}$ is a separator line.  We conclude the following.

\begin{lemma}
We can compute in $O(k^2 n \log n)$ time a set $L$ of $O(k^2)$ lines that contains a separator line.
\label{lem:separator-k2}
\end{lemma}

Let $D_1, D_2$ be a $(2,k)$-center of $P$, let $\ell \in L$ be a line, and let $P^- \subseteq P$ be the set of points that lie in the left halfplane bounded by $\ell$.  We call $D_1, D_2$ a $(2,k)$-center \emph{consistent with $\ell$} if $P^- \cap (D_1 \cup D_2) \subseteq D_1$, the center of $D_1$ lies to the left of $\ell$, and $\partial D_1$ contains at least one point of $P^-$.  
We first describe a decision algorithm that determines whether there is a $(2,k)$-center of unit radius that is consistent with $\ell$.  Next, we describe an algorithm for computing a $(2,k)$-center consistent with $\ell$, which will lead to computing an optimal $(2,k)$-center of $P$, provided there is a well-separated optimal $(2,k)$-center of $P$.  

\paragraph{\textbf{\emph{Decision algorithm.}}}
Let $\ell \in L$ be a line.  We describe an algorithm for determining whether there is a unit radius $(2,k)$-center of $P$ that is consistent with $\ell$.  Let $P^-$ (resp. $P^+$) be the subset of points in $P$ that lie in the left (resp. right) halfplane bounded by $\ell$; set $n^- = |P^-|$, $n^+ = |P^+|$.  Suppose $D_1, D_2$ is a unit-radius $(2,k)$-center of $P$ consistent with $\ell$, and let $c_1, c_2$ be their centers.  Then $P^- \cap (D_1 \cup D_2) \subseteq D_1$ and $|P^- \cap D_1| \geq n^- - k$.
For a subset $Q \subset P$, let $\EuD(Q) = \{D(q) \mid q \in Q\}$ where $D(q)$ is the unit disk centered at $q$.  Let $\EuD^- = \EuD(P^-)$ and $\EuD^+ = \EuD(P^+)$.  For a point $x \in \b{R}^2$, let $\EuD^+_x = \{D \in \EuD^+ \mid x \in D\}$.  Since $\partial D_1$ contains a point of $P^-$ and at most $k$ points of $P^-$ do not lie in $D_1$, $c_1$ lies on an edge of $\EuA_{\leq k}(\EuD^-)$.  

We first compute $\EuA_{\leq k}(\EuD^-)$ in $O(n k \log n)$ time.
For each disk $D \in \EuD^+$, we compute the intersection points of $\partial D$ with the edges of $\EuA_{\leq k}(\EuD^-)$.  By Lemma \ref{lem:k-udc}, there are $O(nk)$ such intersection points, and these intersection points split each edge into \emph{edgelets}.  The total number of edgelets is also $O(nk)$.  
Using Lemma \ref{lem:k-udc}, we can compute all edgelets in time $O(nk\log n)$, because each disk boundary from $\EuD^+$ intersects at most $O(k)$ edges of $\EuA_{\leq k}(\EuD^-)$ and each intersection can be found in $O(\log n)$ time be examining the covering unit disk curves.
All points on an edgelet $\gamma$ lie in the same subset of disks of $\EuD^+$, which we denote by $\EuD^+_{\gamma}$.  Let $P^+_{\gamma} \subseteq P^+$ be the set of centers of disks in $\EuD^+_{\gamma}$, and let $\kappa_{\gamma} = \lambda(\gamma, \EuD^-)$ be the level of $\gamma$ in $\EuD^-$.  A unit disk centered at a point on $\gamma$ contains $P^+_{\gamma}$ and all but $\kappa_{\gamma}$ points of $P^-$.  If at least $k^\prime = k-\kappa_\gamma$ points of $P^+ \setminus P_{\gamma}^+$ can be covered by a unit disk, which is equivalent to $\EuA_{\leq k^{\prime}}(\EuD^+ \setminus \EuD_{\gamma})$ being nonempty, then all but $k$ points of $P$ can be covered by two unit disks.  

When we move from one edgelet $\gamma$ of $\EuA_{\leq k}(\EuD^-)$ to an adjacent one $\gamma^\prime$ with $\sigma$ as their common endpoint, then $\EuD^+_{\gamma} = \EuD^+_{\gamma^\prime}$ (if $\sigma$ is a vertex of $\EuA_{\leq k}(\EuD^-)$), $\EuD^+_{\gamma^\prime} = \EuD^+_{\gamma} \cup \{ D\}$ (if $\sigma \in \partial D$ and $\gamma^\prime \subset \{ D \}$), or $\EuD^+_{\gamma^\prime} = \EuD^+_{\gamma} \setminus \{D\}$ (if $\sigma \in \partial D$ and $\gamma \subset \EuD$).  
We therefore traverse the graph induced by the edgelets of $\EuA_{\leq k}(\EuD)$ and maintain $\EuD_{\gamma}^+$ in the dynamic  data structure described  in Section \ref{sec:arrangementD} as we visit the edgelets $\gamma$ of $\EuA_{\leq k}(\EuD^-)$.  At each step we process an edgelet $\gamma$, insert or delete a disk into $\EuD_{\gamma}^+$, and test whether $\EuA_{\leq j}(\EuD_{\gamma}^+) = \emptyset$ where $j = k - \lambda(\gamma, \EuD^-)$.  If the answer is yes at any step, we stop.  We spend $O(k^3 \log n)$ time at each step, by Lemma \ref{lem:dynD}.  Since the number of edgelets is $O(nk)$, we obtain the following.

\begin{lemma}
Let $P$ be a set of $n$ points in $\b{R}^2$, $\ell$ a line in $L$,  and $0 \leq k \leq n$ an integer.  
We can determine in $O(nk^4 \log n)$ time whether 
there is a unit-radius $(2,k)$-center of $P$ that is consistent with $\ell$.
\end{lemma}

\paragraph{\textbf{\emph{Optimization algorithm.}}}
Let $\ell$ be a line in $L$.  
Let $r^*$ be the smallest radius of a $(2,k)$-center of $P$ that is consistent with $\ell$.  
Our goal is to compute a $(2,k)$-center of $P$ of radius $r^*$ that is consistent with $\ell$.  
We use the parametric search technique \cite{Meg83} --- we simulate the decision algorithm generically at $r^*$ and use the decision algorithm to resolve each comparison, which will be of the form: given $r_0 \in \b{R}^+$, is $r_0 \leq r^*$?  We simulate a parallel version of the decision procedure to reduce the number of times the decision algorithm is invoked.
Note that we need to parallelize only those steps of the simulation that depend on $r^*$, i.e., that require comparing a value with $r^*$.  Instead of simulating the entire decision algorithm, as in ~\cite{Epp97}, we stop the simulation after computing the edgelets and return the smallest $(2,k)$-center found so far, i.e., the smallest radius for which the decision algorithm returned ``yes.''  Since we stop the simulation earlier, we do not guarantee that we find the a $(2,k)$-center of $P$ of radius $r^*$ that is consistent with $\ell$.  However, as argued below this is sufficient for our purpose.  

Let $P^-$, $P^+$ be the same as in the decision algorithm.  Let $\EuD^-$, $\EuD^+$ etc. be the same as above except that each disk is of radius $r^*$ (recall that we do not know the value of $r^*$).  We simulate the algorithm to compute the edgelets of $\EuA_{\leq k}(\EuD^-)$ as follows.  First, we compute the ${\le} k^{th}$ order farthest point Voronoi diagram of $P^-$ in time $O(n \log n + nk^2)$ \cite{AGSS89}.  Let $e$ be an edge of the diagram with points $p$ and $q$ of $P^-$ as its neighbors, i.e., $e$ is a portion of the bisector of $p$ and $q$.  Then for each point $x \in e$, the disk of radius $||xp||$ centered at $x$ contains at least $n^- - k$ points of $P^-$.  We associate an interval $\delta_e = \{||xp|| \mid x \in e\}$.  By definition, $e$ corresponds to a vertex of $\EuA_{\leq k}(\EuD^-)$ if and only if $r^* \in \delta_e$; namely, if $||xp||=r^*$, for some $x \in e$, then $x$ is a vertex of $\EuA_{\leq k}(\EuD^-)$, incident upon the edges that are portions of $\partial D(p)$ and $\partial D(q)$.  
Let $X$ be the sorted sequence of the endpoints of the intervals.
By doing a binary search on $X$ and using the decision procedure at each step, we can find two consecutive endpoints in $X$ between which $r^*$ lies.  We can now compute all edges $e$ of the Voronoi diagram such that $r^* \in \delta_e$.  We thus compute all vertices of $\EuA_{\leq k}(\EuD^-)$.  
Since we do not know $r^*$, we do not have actual coordinates of the vertices.  We represent each vertex as a pair of points.  
Similarly, each edge is represented as a point $p \in P^-$, indiciating that $e$ lies in $\partial D(p)$, and it can be computed using the cells of the Voronoi diagram.  Given a vertex of $\EuA_{\leq k}(\EuD^-)$ and an outgoing edge, represented by the point $p \in P^-$, we can compute the other endpoint as the next edge $e^\prime$ of the Voronoi cell of the $p$ that is a point in $\EuA_{\leq}(\EuD^-)$ by walking around the boundary of the cell.  
Once we have all the edges of $\EuA_{\leq k}(P^-)$, we can construct the graph induced by them and compute $O(k^2)$ $x$-monotone unit-disk curves whose union is the set of edges in $\EuA_{\leq k}(P^-)$, using Lemma \ref{lem:udcs-k2}.  Since this step does not depend on the value of $r^*$, we need not parallelize it.  Let $\Xi = \{\xi_i, \ldots, \xi_u\}$, $u=O(k^2)$, be the set of these curves.  

Next, for each disk $D \in \EuD^+$ and for each $\xi_i \in \Xi$, we compute the edges of $\xi_i$ that $\partial D$ intersects, using a binary search.  We perform these $O(nk^2)$ binary searches in parallel and use the decision algorithm at each step.  Incorporating Cole's technique~\cite{Col87} in the binary search, the decision procedure is invoked only $O(\log n)$ times.  For an edge $e \in \EuA_{\leq k}(\EuD)$, let $\EuD_e^+ \in \EuD$ be the set of disks whose boundaries intersect $e$.  We sort the disks in $\EuD_e^+$ by the order in which their boundaries intersect $e$.  By doing this in parallel for all edges and using a parallel sorting algorithm for each edge, we can perform this step by invoking the decision algorithm $O(\log n)$ times.  
The total time spent is $O(nk^4 \log^2 n)$.



\paragraph{\textbf{\emph{Putting pieces together.}}}
We repeat the optimization algorithm for all lines in $L$ and return the smallest $(2,k)$-center that is consistent with a line in $L$.  
Since Lemma \ref{lem:separator-k2} shows that as long as the solution is well separated at least one line in $L$ is a separator line for the optimal $(2,k)$-center of $P$, the smallest radius returned must be that of the optimal $(2,k)$-center of $P$.
Hence, we conclude the following:

\begin{lemma}
Let $P$ be a set of $n$ points in $\b{R}^2$ and $0 \leq k \leq n$ an integer.  
If an optimal $(2,k)$-center of $P$ is well separated, then the $(2,k)$-center problem for $P$ can be solved in $O(n k^6 \log^2 n)$ time. 
\label{lem:far}
\end{lemma}

\section{Nearly Concentric Disks}
\label{sec:near}

\begin{figure}[ht]
  \centering
  \includegraphics{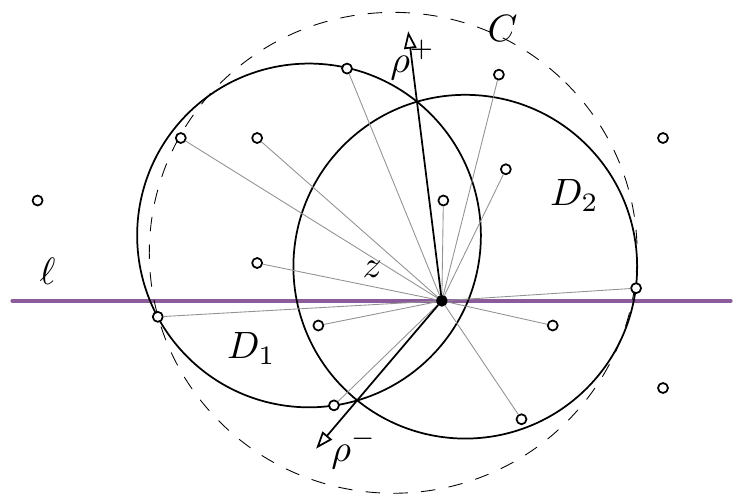}
  \caption{\label{fig:near}
          Two unit disks $D_1$ and $D_2$ or radius $r^*$ with centers closer than a distance $r^*$.  
      }
\end{figure}

In this section we describe an algorithm for the case in which the two disks $D_1$ and $D_2$ of the optimal solution are not well separated.  More specifically, let $c_1$ and $c_2$ be the centers of $D_1$ and $D_2$ and let $r^*$ be their radius.  This section handles the case where $||c_1 c_2|| \leq r^*$.  

First, we find an \emph{intersector point} $z$ of $D_1$ and $D_2$ --- a point that lies in $D_1 \cap D_2$.  
We show how $z$ defines a set $\EuP$ of $O(n^2)$ possible partitions of $P$ into two subsets, such that for one partition $P_{i,j}$, $P \setminus P_{i,j}$ the following holds: $(D_1 \cup D_2) \cap P = (D_1 \cap P_{i,j}) \cup (D_2 \cap (P \setminus P_{i,j}))$.
Finally, we show how to search through the set $\EuP$ in $O(k^{7} n^{1+\delta})$ time, deterministically, for any $\delta>0$, or in $O(k^{7} n \log^3 n)$ expected time.  

\paragraph{\textbf{\emph{Finding an intersector point.}}}
Let $C$ be the circumcircle of $P \cap (D_1 \cup D_2)$.
Eppstein~\cite{Epp97} shows that we can select $O(1)$ points inside $C$ such that at least one, $z$, lies in $D_1 \cap D_2$.  
We can hence prove the following.

\begin{lemma}
Let $P$ be a set of $n$ points in $\b{R}^2$.  We can generate in $O(nk^3)$ time a set $\EuZ$ of $O(k^3)$ points such that for any nearly concentric $(2,k)$-center $D_1, D_2$, one of the points in $\EuZ$ is their intersector point.  
\end{lemma}

\begin{proof}
%
Using Matou\v{s}ek's~\cite{Mat95} algorithm for solving LP-type problems with violations, in $O(k^3 n)$ time we can find the smallest circle that contains $n-k$ points of $P$.  
Briefly, the algorithm runs by finding the three points defining the circumcircle, removing each one in turn, and recursing until $k$ points have been removed.  Matou\v{s}ek shows that if we keep track of which nodes in the recursion we reach and halt the recursion if we have seen that node before, then the size of the recursion tree is only $O(k^3)$.  
In the running of this algorithm we generate all circles which include exactly $n-j$ points of $P$ for $0 \leq j \leq k$.  We claim that one of these circles must be $C$. 

If the initial circle is not $C$, then it must have at least one point on its 
boundary which is not in $P \cap (D_1 \cup D_2)$.  At least one path of the recursion removes this point.  Since we can reach the point set $P \cap (D_1 \cup D_2)$ in at most $k$ steps, some step in this recursion must return $C$.  

Finally, since the area of $D_1 \cup D_2$ is a constant fraction of $C$ when $D_1, D_2$ are nearly concentric, then by selecting a constant number of points in $C$ one can be guaranteed to be an intersector point. 
\end{proof}


Let $z$ be an intersector point of $D_1$ and $D_2$, and let $\rho^+$, $\rho^-$ be the two rays from $z$  to the points of $\partial D_1 \cap \partial D_2$.  
Since $D_1$ and $D_2$ are nearly concentric, the angle between them is at least some constant $\theta$.  We choose a set $U \subseteq S^1$ of $h = \lceil 2 \pi / \theta \rceil$ uniformly distributed directions.  For at least one $u \in U$, the line $\ell$ in direction $u$ and passing through $z$ separates $\rho^+$ and $\rho^-$, see Figure \ref{fig:near}.  We fix a pair $z,u$ in $Z \times U$ and compute a $(2,k)$-center of $P$, as described below.  We repeat this algorithm for every pair.  If $D_1$ and $D_2$ are nearly concentric, then our algorithm returns an optimal $(2,k)$-center.  

\paragraph{\textbf{\emph{Fixing $z$ and $u$.}}}
For a subset $X \subset P$ and for an integer $t\geq 0$, let $r^t(X)$ denote the minimum radius of a $(1,t)$-center of $X$.  Let $P^+$ (resp. $P^-$) be the subset of $P$ lying above (resp. below) the $x$-axis; set $n^+ = |P^+|$ and $n^- = |P^-|$.  Sort $P^+ = \langle p_1^+, \ldots, p_{n^+}^+\rangle$ in clockwise order and $P^- = \langle p_1^-, \ldots, p_{n^-}^- \rangle$ in counterclockwise order.  
For $0 \leq i \leq n^+$, $0 \leq j \leq n^-$, let $P_{i,j} = \{p_1^+, \ldots, p_i^+, p_1^-, \ldots, p_j^- \}$ and $Q_{i,j} = P \setminus P_{i,j}$.  
For $0 \leq t \leq k$, let 
$$
m_{i,j}^t = \max \{ r^t(P_{i,j}), r^{k-t}(Q_{i,j}) \}.
$$
For $0 \leq t \leq k$, we define an $n^+ \times n^-$ matrix $M^t$ such that $M^t(i,j) = m^t_{i,j}$.  

Suppose $z$ is an intersector point of $D_1$ and $D_2$, $\ell$ separates $\rho^+$ and $\rho^-$, and $\rho^+$ (resp. $\rho^-$) lies between $p_a^+, p_{a+1}^+$ (resp. $p_b^-, p_{b+1}^-$).  Then $P \cap (D_1 \cup D_2) = (P_{a,b} \cap D_1) \cup (Q_{a,b} \cup D_2)$; see Fig \ref{fig:near}.  
If $|P_{a,b} \setminus D_1| = t$, then $r^* = m^t_{a,b}$.  
The problem thus reduces to computing 
$$
\mu(z,u) = \min_{i,j,t} m^t_{i,j}
$$
where the minimum is taken over $0 \leq i \leq n^+$, $0 \leq j \leq n^-$, and $0 \leq t \leq k$.  For each $t$, we compute $\mu^t(z,u) = \min_{i,j} m^t_{i,j}$ and choose the smallest among them.  

We note two properties of the matrix $M^t$ that will help search for $\mu^t(z,u)$:

\begin{itemize}
\item{(P1)} 
If $r^t(P_{i,j}) > r^{k-t}(Q_{i,j})$ then $m^t_{i,j} \leq m^t_{i^\prime,j^\prime}$ for $i^\prime \geq i$ and $j^\prime \geq j$.  These partitions only add points to $P_{i,j}$ and removes points from $Q_{i,j}$, and thus cannot decrease $r^t(P_{i,j})$ or increase $r^{k-t}(Q_{i,j})$.  
Similarly, if $r^{k-t}(Q_{i,j}) > r^t(P_{i,j})$, then $m^t_{i,j} < m^t_{i^\prime, j^\prime}$ for $i^\prime \leq i$ and $j^\prime \leq j$. 
\item{(P2)} Given a value $r$, if $r^t(P_{i,j}) > r$, then $m^t_{i^\prime,j^\prime} > r$ for $i^\prime \geq i$ and $j^\prime \geq j$, and if $r^t(Q_{i,j}) > r$, then $m^t_{i^\prime,j^\prime} > r$ for $i^\prime \leq i$ and $j^\prime \leq j$.
\end{itemize}

\paragraph{\textbf{\emph{Deterministic solution.}}}
We now have the machinery to use a technique of Frederickson and Johnson \cite{FJ82}.  
For simplicity, let us assume that $n^+ = n^- = 2^{\tau+1}$ where $\tau  = \lceil \log_2 n \rceil + O(1)$.  The algorithm works in $\tau$ phases.
In the beginning of the $h$th phase we have a collection $\EuM_h$ of $O(2^h)$ submatrices of $M^t$, each of size $(2^{\tau - h+1} + 1) \times (2^{\tau-h+1} + 1)$.  Initially $\EuM_1 = \{M^t\}$.  In the $h$th phase we divide each matrix $N \in \EuM_h$ into four submatrices each of size $(2^{\tau-h} + 1) \times (2^{\tau-h} + 1)$ that overlap along one row and one column.  We call the cell common to all four submatrices the \emph{center cell} of $N$.  Let $\EuM_h^\prime$ be the resulting set of matrices.  Let $\EuC = \{ (i_1, j_1), \ldots, (i_s, j_s)\}$ be the set of center cells of matrices in $\EuM_h$.  We compute $m^t_{i_l, j_l}$ for each $1 \leq l \leq s$.  We use (P1) to remove the matrices of $\EuM_h$ that are guaranteed not to contain the value $\mu^t(z,u)$.  
In particular, if $m^t_{i_l,j_l} = r^t(P_{i_l,j_l})$ and there is a matrix $N \in \EuM^\prime_h$ with the upper-left corner cell $(i^\prime, j^\prime)$ such that $i^\prime \leq i_l$ and $j^\prime \leq j_l$, then we can remove $N$.  
Similarly if $m^t_{i_l,j_l} = r^{k-t}(Q_{i,j})$ and there is a matrix $N \in \EuM^\prime_h$ with the lower-right corner cell $(i^\prime,j^\prime)$ such that $i^\prime \geq i_l$ and $j^\prime \geq j_l$, we can delete $N$. 
We then set $\EuM^\prime_h$ to $\EuM_{h+1}$. 

\begin{figure}[h]
\center
\includegraphics{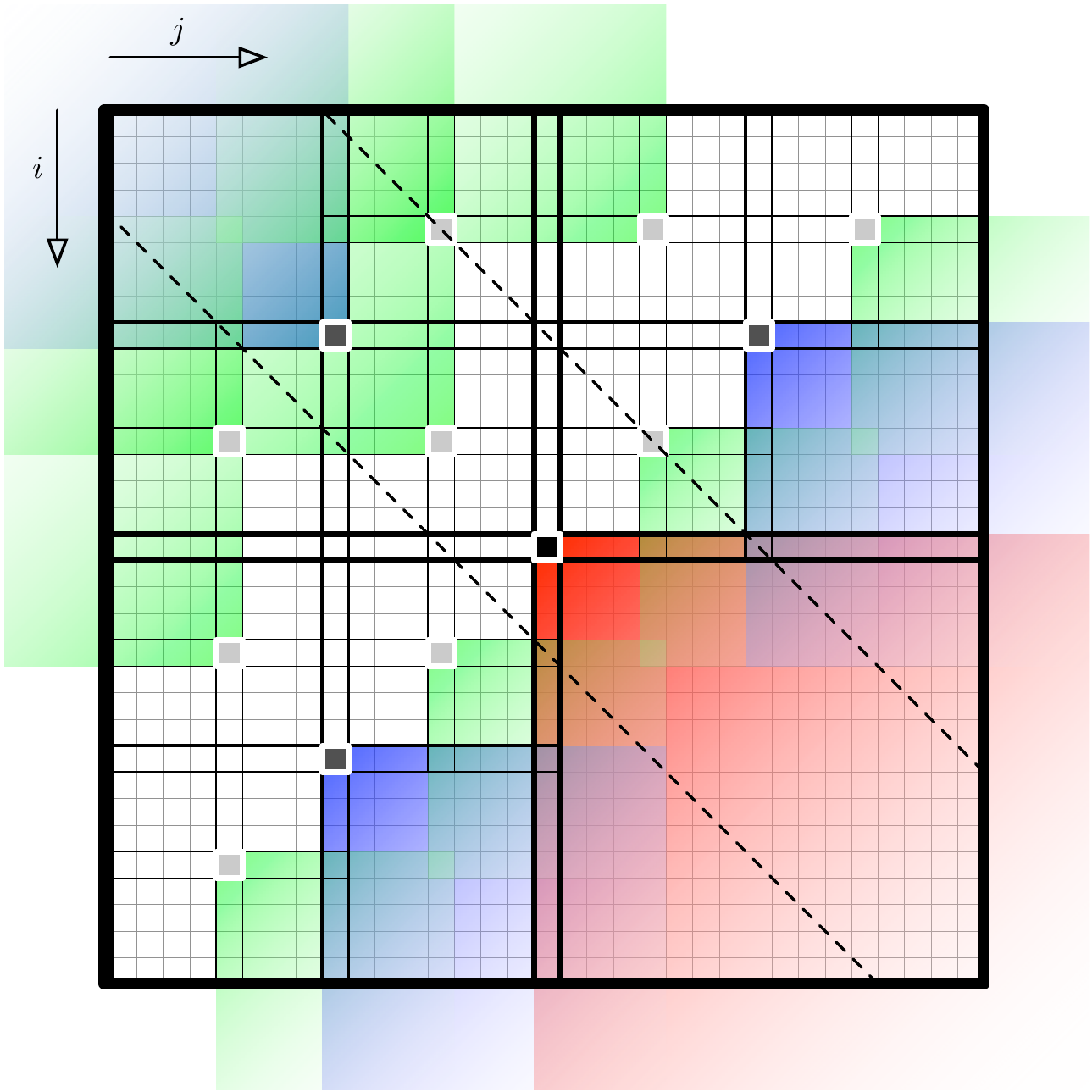}
\caption{\label{fig:FJ-mat}
Example of running deterministic algorithm through 3 phases.  Shaded regions have been pruned.  Center cells are darkened.  
}
\end{figure}

\begin{lemma}
Before the $h$th phase consider a diagonal from large $i$ and $j$ to small $i$ and $j$ that passes through at least one center cell of a matrix $N \in \EuM_h$.  It passes through at most one more center cell of a matrix $N^\prime \in \EuM_h$.  
\label{lem:diag-center}
\end{lemma}

\begin{proof}
We show this inductively.  The base case is clearly true for the single center cell in $\EuM_1$.  
Assume it is true for $\EuM_h$, then we show it is true for $\EuM_{h+1}$.  
See Figure \ref{fig:FJ-mat}.
We consider two cases, first the diagonal passes through a center cell of $\EuM_h$.  
In this case if it passes through two center cells of $\EuM_h$, then it passes through $4$ center cells of $\EuM^\prime_h$, but the pruning step eliminates at least two of them.  
In the second case, the diagonal does not pass through a center cell of $\EuM_h$.  We can bound the number of center cells of matrices it passes through in $\EuM^\prime_h$ to $4$ using the inductive hypothesis.  
Consider one of the interior center cells $(i,j) \in N \in \EuM^\prime_h$ it passes through, neither the first not the last.  When the pruning step for the matrix in $\EuM_h$ that contains $N$ is called, it either eliminated the other matrixes in $\EuM_h^\prime$ that the diagonal passes to before or after $N$.  If the diagonal passes through 3 center cells in $\EuM_h^\prime$, then this reduces it to two, if the diagonal passes through 4 center cells, the applying this analysis to both interior matrices reduces it to two.  
\end{proof}

Lemma \ref{lem:diag-center} implies that $O(n)$ cells remain in $\EuM^\prime_h$ after the pruning step and that they can be connected by two monotone paths in $\EuM^t$, which consists of $O(n)$ cells.
Since $P_{i,j}$ differs from $P_{i-1,j}$ and $P_{i,j-1}$ by one point, we can compute $m^t_{i_l,j_l}$ for all $(i_l,j_l) \in \EuC$ using 
Lemma \ref{lem:dynC} in total time $O(k^3 n^{1+\delta})$.
%
%
Hence, each phase of the algorithm takes $O(k^3 n^{1+\delta})$ time.  

\begin{lemma}
Given $z \in Z$, $u \in U$, and $0 \leq t \leq k$, 
$\mu^t(z,u)$ can be computed in time $O(k^3 n^{1+\delta})$, for any $\delta > 0$.  
\end{lemma}

\paragraph{\textbf{\emph{Randomized solution.}}}
We can slightly improve the dependence on $n$ by using the dynamic data structure in Section \ref{sec:arrangementD} and (P2).  
%
As before, in the $h$th phase, for some constant $c>1$, we maintain a set $\EuM_{h}$ of at most $c2^h$ submatrices of $M^t$, each of side length $2^{\tau-h+1}+1$, and their center cells $\EuC$.  Each submatrix is divided into four submatrices of side length $2^{\tau-h}+1$, forming a set $\EuM^\prime_{h}$.  To prune $\EuM^\prime_{h}$, we choose a random center cell $(i,j)$ from $\EuC$ and evaluate $r = m^t_{i,j}$ in $O(k^3 n)$ time.  For each other center cell $(i^\prime, j^\prime) \in \EuC$, $m^t_{i^\prime,j^\prime} >r$ with probability $1/2$, and using (P2), we can remove a submatrix from $\EuM^\prime_{h}$. 
More specifically, if $m^t_{i^\prime,j^\prime}>r$, then any matrix $N \in \EuM^\prime_h$ with an lower right corner $(i^\prime, j^\prime)$ such that $i^\prime \leq i$ and $j^\prime \leq j$ or a upper left corner $(i^{\prime\prime}, j^{\prime\prime})$ such that $i^{\prime\prime} \geq i$ and $j^{\prime\prime} \geq j$, then we can prune $N$ from $\EuM_h^\prime$.  
Eppstein~\cite{Epp97} proves that by repeating this process a constant number of times, we expect to reduce the size of $\EuM^\prime_h$ to $c2^{h+1}$.


On each iteration we use the dynamic data structure described in Section \ref{sec:arrangementD}.  For $O(n)$ insertions and deletions, it can compare each center cell from $\EuC$ to $r$ in $O(k^{3} n \log^2 n)$ time.  
Thus, finding $\mu^t(z,u)$ takes expected $O(n k^{3} \log^3 n)$ time.  

\begin{lemma}
Given $z \in Z$, $u \in U$, and $0 \leq t \leq k$, 
$\mu^t(z,u)$ can be computed in expected time $O(k^3 \log^3 n)$.
\label{lem:nearR}
\end{lemma}


\paragraph{\textbf{\emph{Putting pieces together.}}}
By repeating either above algorithm for all $0 \leq t \leq k$ and for all pair $(z,u) \in Z \times U$, we can compute a $(2,k)$-center of $P$ that is optimal if $D_1$ and $D_2$ are nearly concentric.  Combining this with Lemma \ref{lem:far}, we obtain the main result of the paper.

\begin{theorem}
Given a set $P$ of $n$ points in $\b{R}^2$ and an integer $k \geq 0$, an optimal $(2,k)$-center of $P$ can be computed in $O(k^7 n^{1+\delta})$ (deterministic) time, for any $\delta>0$ 
or in $O(k^7 n \log^3 n)$ expected time.
\end{theorem}

\section{The $(p,k)$-Center Problem Under the $\ell_\infty$ Metric}

This section focuses on the $\ell_\infty$ version of the $(p,k)$-center problem, and hence all references to the $(p,k)$-center problem herein are referring to the $\ell_\infty$ variant.  We use extensively that for $p \leq 3$, the $(p,0)$-center problem is LP-type~\cite{SW96}, and thus for $p\leq 3$ the $(p,k)$-center problem can be solved in $O(k^{O(1)} n)$ time.  We also use that if all points lie in $\b{R}^1$, then the $(p,0)$-center problem is LP-type for any $p>0$, with combinatorial dimension $O(p)$, and thus in $\b{R}^1$, the $(p,k)$-center problem can be solved in $O(k^{O(p)}n)$ time.  

Like in the $\ell_2$ variant, we first study the decision version of the dual problem; here an arrangement of unit squares.  
Let $\EuS = \{S_1, \ldots, S_n\}$ be a set of $n$ unit squares (side length 1) in $\b{R}^2$.  Let $\EuA(\EuS)$ be the arrangement of $\EuS$.  We say a point $q$ \emph{stabs} a square $S \in \EuS$ if $q \in S$.  Let $\EuS(q) \subset \EuS$ be the set of squares stabbed by $q$.  

We seek to determine whether there exists a placement of $p$ points $\EuQ = \{q_1, \ldots, q_p\}$ such that $\left| \bigcup_{q \in \EuQ} \EuS(q) \right| \geq n-k$.  
We refer to this as the $(p,k)$-stabbing decision problem.  
All of our algorithms also return a solution if one exists.  
By replacing each point in the $(p,k)$-center problem with a unit square centered at that point, then the $p$ stabbing points of the $(p,k)$-stabbing decision problem serve as the center points of unit squares that contain $n-k$ of the original point set.

\paragraph{Structure.}
We start by reviewing structure observed by Sharir and Welzl~\cite{SW96} about the $(p,0)$-stabbing decision problem.  

If a horizontal or vertical line $\ell$ passes through all $S \in \EuS$, then this $(p,0)$-center decision problem reduces to a variant in $\b{R}^1$ because any stabbing point $q$ can be replaced with $q^\prime$, the closest point on $\ell$ to $q$, so that $\EuS(q) \leq \EuS(q^\prime)$.  We can then solve the $(p,0)$-stabbing decision problem in $O(n)$ time or the $(p,k)$-center problem in $O(n k^{O(p)})$ time.  
We henceforth assume that this is not the case.  

Let $\ell^L$ describe the line passing through the right boundary of the leftmost square.  Similarly, let $\ell^R$ (resp. $\ell^T$, $\ell^B$) describe the line passing through the left (resp. bottom, top) boundary of the rightmost (resp. topmost, bottommost) square.  
Let $H_0$ describe the rectangle bounded on its left side by $\ell^L$, its right side by $\ell^R$, its bottom side by $\ell^B$, and its top side by $\ell^T$.  (See Figure \ref{fig:rect4}.)
$H_0$ must have positive area otherwise a horizontal or vertical line would pass through the set of all squares.  

Let $H_0 \cap \ell^X$, for $X \in \{L,R,T,B\}$, describe the four boundary segments of $H_0$.  Call the intersection of two boundary segments a corner of $H_0$.  
If the $(p,0)$-stabbing decision problem has a solution, we claim that each boundary segment of $H_0$ contains a stabbing point in a solution of the $(p,0)$-stabbing decision problem (in particular, the solution of $p$ stabbing points contained in the smallest rectangle).  For instance, if $H_0 \cap \ell^L$ does not contain a stabbing point, then we can replace the point $q$ stabbing the leftmost square with another point $q^\prime$ on $H_0 \cap \ell^L$ such that $\EuS(q) \leq \EuS(q^\prime)$.  

If a stabbing point $q$ lies on corner, it lies on two boundary segments at once, and we can set $\EuS^\prime = \EuS \setminus \EuS(q)$ and then solve the $(p-1,k)$-stabbing decision problem on $\EuS^\prime$.  Of course, we don't know which corner is a stabbing point, but there are a constant number and we can try them all.  

Define $\ell^L_{j^L}$ (resp. $\ell^R_{j^R}$, $\ell^T_{j^T}$, $\ell^B_{j^B}$) as the line through the right (resp. left, bottom, top) boundary of the $j^L$th leftmost (resp. $j^R$th rightmost, $j^T$th topmost, $j^B$th bottommost) square.  
We can also define the rectangle $H_{j^L, j^R, j^T, j^B}$ which is defined by the intersection of halfspaces defined by lines $\ell^L_{j^L}$, $\ell^R_{j^R}$, $\ell^T_{j^T}$, and $\ell^B_{j^B}$.  
We actually want to be slightly careful since one square may be in the $j^L$th leftmost and $j^T$th topmost squares.  We count squares first from left and right, then those remaining from top and bottom.  
Let $\EuS_{j^L, j^R, j^T, j^B}$ be the set of squares which intersect $H_{j^L, j^R, j^T, j^B}$.

\paragraph{Dynamic data structure.}
We will need a data structure to be able to maintain $H_0$ and $\EuS$ under the removal of the set $\EuS(q)$ for a possible stabbing point $q$.  Sharir and Welzl \cite{SW96} provide a data structure that stores a set of canonical subsets, such that under this operation $\EuS \setminus \EuS(q)$ can be stored as the union of $O(\log n)$ (not necessarily disjoint) canonical subsets.  The new boundary lines of $H_0$ can be constructed in $O(\log n)$ time from the $O(\log n)$ subsets.  

The structure is built, and extended to handle outliers, as follows.
In the $x$- and $y$-directions store binary trees of $\EuS$ sorted by their coordinates.  Each node in the tree stores a canonical subset of all squares in its subtree.  For a query point $q$, we can return all squares that cannot intersection $q$ based on $x$- and $y$-coordinates independently, as a set of $O(\log n)$ canonical subsets each.  The union is $\EuS \setminus \EuS(q)$.  
We may need to build this data structure $p-1$ levels deep on each canonical subset for solving the $(p,k)$-stabbing decision problem.  
To construct $H_0$ quickly, we can find the maximum and minimum square in $x$ and $y$ coordinate over all $O(\log n)$ canonical subsets.  To instead construct $H_{j^L, j^R, j^T, j^B}$, we can find the $j^L$ minimum $x$ coordinate in $O(j^L \log n)$ time and similarly for $j^R$, $j^T$, and $j^B$; thus constructing $H_{j^L, j^R, j^T, j^B}$ can be done in $O(k \log n)$ time, where $j^L, j^R, j^T, j^B \leq k$.

\subsection{The $(4,k)$-Stabbing Decision Problem}

First we choose positive integral values $j^L$, $j^R$, $j^T$, and $j^B$ such that $j^L + j^R + j^T + j^B \leq k+4$ and create $H_{j^L, j^R, j^T, j^B}$.  
If $j^L$, $j^R$, $j^T$, and $j^B$ are chosen correctly, then $H_{j^L, j^R, j^T, j^B}$ is the smallest rectangle that contains the 4 stabbing points.  If the decision is true, then one of this set of $O(k^4)$ rectangles must match the solution because it can not exclude more than $k$ rectangles in any one direction.
In what follows, we assume we have chosen $j^L, j^R, j^T, j^B$ correctly, but in the full algorithm we try each until we find a solution.  
If $H_{j^L, j^R, j^T, j^B}$ has non positive area then we can solve the problem in $\b{R}^1$.  We then see if one of the corners, $q$, of $H_{j^L, j^R, j^T, j^B}$ can be a stabbing point by solving the $(3,k-(j^L+j^R + j^T + j^B-4))$-stabbing decision problem on $\EuS_{j^L, j^R, j^T, j^B} \setminus \EuS(q)$.  If the answer is negative for each corner, and we assume that we have chosen $j^L$, $j^R$, $j^T$, $j^B$ correctly, then each boundary segment of $H_{j^L, j^R, j^T, j^B}$ must contain a distinct stabbing point.  
Let $\EuS^I \subset \EuS_{j^L, j^R, j^T, j^B}$ be the subset so that each $S \in \EuS^I$ does not intersection $\partial H_{j^L, j^R, j^T, j^B}$ --- these squares must be totally contained in $H_{j^L, j^R, j^T, j^B}$.  Let $k^I = |\EuS^I|$.  

In the following we assume that $H_{j^L, j^R, j^T, j^B}$ is the smallest rectangle to contain all stabbing points and to simplify notation we set $\EuS^\prime = \EuS_{j^L, j^R, j^T, j^B} \setminus \EuS^I$, $\kappa = k-(j^L+j^R + j^T + j^B-4) - k^I$, and $H = H_{j^L,j^R,j^T,j^B}$.  Finally, we assume that the solution to the $(p,\kappa)$-stabbing decision problem on $\EuS^\prime$ has no point on the corners of $H$.

\begin{figure}[h!!t]
  \centering
  \includegraphics{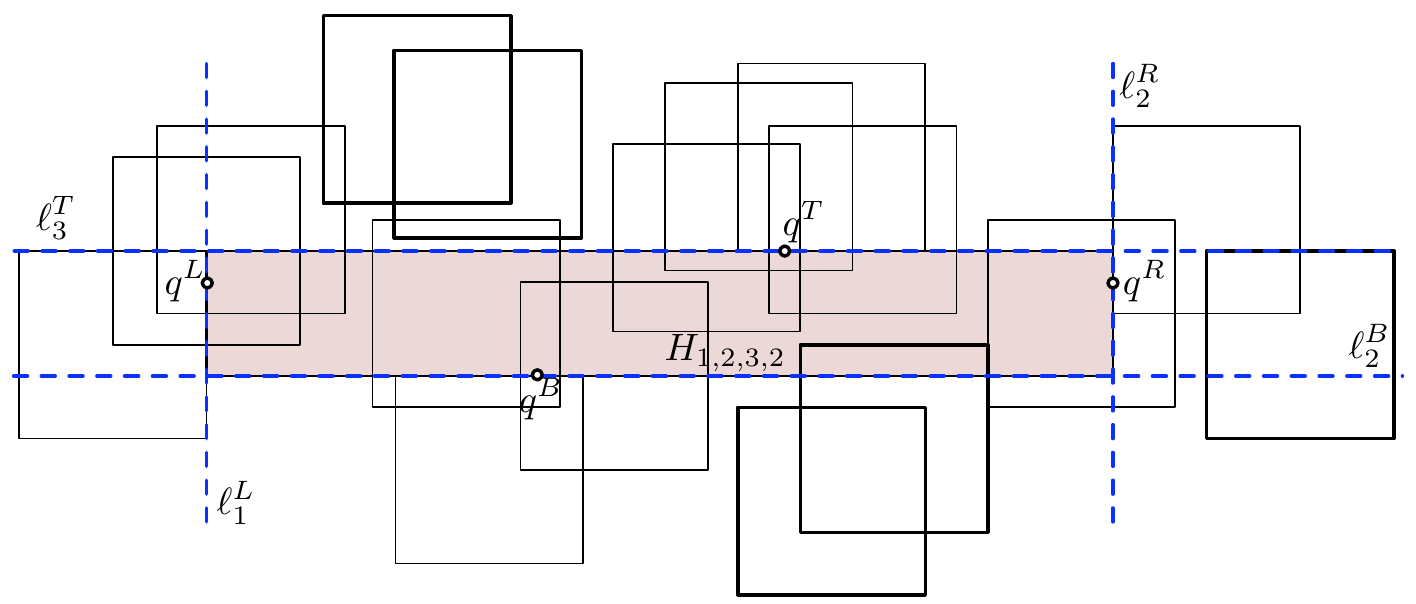}
  \caption{\label{fig:rect4}
                Structure of a $4$-center problem with $\ell_\infty$-distance.  Rectangle $H = H_{1,2,3,2}$ is shaded and bounded by lines $\ell^T, \ell^R, \ell^B, \ell^L$ on the top, right, bottom, and left sides, respectively.  The four centers appear on the four sides of $H$ labeled.  There are five outliers squares shown in bold.}
\end{figure}

\paragraph{4 Rotating Calipers.}
We can now apply a rotating calipers type technique with four calipers, with one point on each edge of $H$.  Since each square can intersect each edge of $H$ at most twice, the boundary of $H$ is divided into $O(n)$ regions such that all points within a region of the boundary intersect the same set of squares.  Squares can intersect more than one edge of $H$, either by also containing a corner point (i.e. left and top), intersect two opposite sides (i.e. left and right), or both (i.e. left, top, and right).  In the third case when a square intersects three sides it must entirely contain one of those sides, and thus any point chosen on that side must stab that square and we can ignore it.  Also only one pair, w.l.o.g. top and bottom, can have squares intersecting both, otherwise both pairs of opposite sides are shorter than a distance $1$, and any square intersecting a pair of opposite sides must entirely contain one of the other sides.  Assuming the top and bottom edges are longer than 1 (so no square can intersect both the left and right edge) we consider two cases: where the point on the top side is right of the point on the bottom side, and vice versa.  We focus on the first case and handle the other one symmetrically.  

We treat the subset of squares $\EuS_2 \subseteq \EuS^\prime$ which intersect the top and bottom edges separately from the subset $\EuS_E = \EuS^\prime \setminus \EuS_2$ of the ones that only intersection only one edge or two adjacent edges.  
Each square $S \in \EuS_E$ describes one interval on the curve defined by $\partial H$.  Thus, given a placement of four stabbing points, one on each boundary side, the squares from $\EuS_E$ which are not stabbed lie in one of four intervals of $\partial H$ bounded by the stabbing points.  In the optimal solution let there be $i^R$ unstabbed squares in $\EuS_E$ between the $q^R$ and $q^T$, $i^T$ squares between $q^T$ and $q^L$, $i^L$ squares between $q^L$ and $q^B$, and $i^B$ squares between $q^B$ and $q^R$.  For any values $i^L$, $i^R$, $i^T$, and $i^B$ we can determine if there is placement of the stabbing points on $\partial H$ that has exactly those many unstabbed squares in the associated intervals.  
Given a placement of $q^R$ in bottommost region of the right boundary edge, we can try to place $q^T$ skipping $i^R$ squares, then place $q^L$ skipping $i^T$ squares, and finally $q^B$ skipping $i^L$ squares.  If there are $i^B$ squares remaining it is successful.  If it is not successful at any placement step, then we shift $q^R$ to the next region up on the right boundary and try shifting the other stabbing points to the next region in a counter-clockwise direction to satisfy the constraints.  If all attempts are unsuccessful for all placements of $p^R$ on the right edge, then this choice of $i^R$, $i^L$, $i^B$, and $i^T$ is incorrect.  Since there are only $O(n)$ regions, and each stabbing point is in each region at most once, since they only move counter-clockwise, this takes $O(n)$ time.

Once a solution for $\EuS_E$ has been found, we attempt find a solution for $\EuS_2$.  These squares can be sorted left to right and a successful stabbing will have $i^1$ unstabbed squares from $\EuS_2$ left of $q^B$, $i^2$ squares between $q^B$ and $q^T$, and $i^3$ squares right of $q^T$, for some nonnegative integers $i^1, i^2, i^3$.  We can now adjust $q^B$ and $q^T$ such that the sets $\EuS_E(q^B)$ and $\EuS_E(q^T)$ do not change.  The boundaries of the squares from $\EuS_2$ divide the regions into intervals so that within a interval $\EuS_2(q^B)$ and $\EuS_2(q^T)$ do not change.  After preprocessing to find the left boundary of the rightmost square in $\EuS_2$ and the right boundary of the leftmost square in $\EuS_2$, in $O(i^1 + i^3)$ time we check if we can place $q^B$ and $q^T$ to satisfy $i^1$ and $i^3$.  
Quickly checking the $i^2$ constraint requires preprocessing on the intervals created by the sorted ordering of $\EuS_2$ so each region contains the number of points stabbed and the number of unstabbed squares to the right.  Thus if $q^T$ is in a region so that it stabs $s$ squares and there are $r$ squares to the right of the region that $q^B$ is in, then there are $r - s - i^3$ unstabbed squares from $\EuS_2$ between $q^B$ and $q^T$.  If $r-s-i^3 \leq i^2$ then we return true, if not we go back to dealing with $\EuS_E$ and shift the stabbing points in counter-clockwise order.  

Although, we do not know the values of $i^R$, $i^L$, $i^T$, $i^B$, $i^1$, $i^2$, and $i^3$ we do know that $i^R + i^L + i^T + i^B + i^1 + i^2 + i^3 = \kappa$, thus there are only $O(k^6)$ possible values.  
For each set of values, we require $O(n)$ time to handle $\EuS_E$ and for each step $O(k)$ time to handle $\EuS_2$, after preprocessing.   
Let $T_\infty(n, p,k)$ be the required time for the algorithms described above to solve the $(p,k)$-stabbing decision problem on $n$ unit squares.  

\begin{lemma}
$T_\infty(n, 4, k) = O(k^4(T_\infty(n,3,k) + k^{11} n) + n \log n)$ or just $O(k^4(T_\infty(n,3,k) + k^{11} n))$ if the squares are presorted along the $x$- and $y$-axis.  
\end{lemma}

\begin{theorem}
$T_\infty(n, 4, k) = O(k^{O(1)} n + n \log n)$ or just $O(k^{O(1)}n)$ if the squares are presorted along the $x$- and $y$-axis.  
\label{thm:rect4-dec}
\end{theorem}

\subsection{The $(5,k)$-Stabbing Decision Problem}
We first construct $O(k^4)$ rectangles $H = \partial H_{j^L, j^R, j^T, j^B}$ as above.  To simplify notation, also assume that $k$ squares intersect or lie inside of $H$ and that at least one center must lie on each side of $H$.  We now have to consider $3$ cases.  

First, one of the centers lies on a corner of $H$.  In this case, we can try all corners, remove the squares that intersect that corner and apply the algorithm for $p=4$ on the remaining squares.  

Second, all $5$ of the centers lie on the rectangle $H$ (not its interior), but none lie on a corner.  
Third, $4$ centers lie on the boundary of $H$, but none lie on a corner and the fifth center lies in the interior of $H$.  
These cases are more complicated and requires the dynamic data structure described above.

In the second and third case we choose non-negative integers $i_1$ through $i_9$ such that $i_1 + i_2 + i_3 + i_4 + i_5 + i_6 + i_7 + i_8 + i_9 = k$.  These determine which points are outliers and not contained in the $5$ centers.  We choose $O(k^8)$ sets of integers and complete the following for each set.  For what follows we assume we have chosen the correct set.  
Each side of $H$ has at least one center, and one side has two.  We perform the following, assuming each side, in turn, has two centers; w.l.o.g. let it be the right side.
Guess that the left most interval on the bottom edge of $H$ contains the center point, $p^B$.  
Let $\EuS_{p^B} = \EuS \setminus \EuS(p^B)$ be the set of squares that do not contain $p^B$.  Using the above 4-level dynamic data structure, obtain $\EuS_{p^B}$ and construct $H^\prime = H_{1, 1, 1, 1+i_1}$ on $\EuS_{p^B}$.  The bottommost $i_1$ remaining squares have been designated as outliers, not to contain any center point.  
Now either the bottom left or the bottom right corner of $H^\prime$ must contain a center point.  
Check each case by the following; w.l.o.g. assume its the bottom right corner, $p^R$.  Create $\EuS_{p^B} \setminus \EuS_{p^B}(p^R)$ using the dynamic data structure, and recalculate $H_{1, 1+i_2, 1, 1+i_3}$.  
Again the $i_2$ rightmost squares and $i_3$ bottommost squares are designated outliers.  
Now again either the bottom right or bottom left corner of $H$ must be a center point.  Check either, remove $i_4$ and $i_5$ outliers, and proceed as before removing the squares contained in the third center.  
This process repeats once more, removing $i_6$ and $i_7$ outliers, and squares containing the fourth center.  There is now one center left to place.  We remove $i_8$ and $i_9$ outliers and can easily check if the last center can contains all remaining squares.  If it cannot then we update $p^B$ by sliding it to the next interval on the bottom edge of $H$.  We update our 4-level data structure in $O(\log^4 n)$ time.  This repeats until either all squares can be stabbed by the last center, meaning the result is true, or all intervals on the bottom edge of $H$ have been tried, meaning the result is false.  

Accounting for the $O(k^4)$ possible outliers to create the initial rectangle $H$, and the $O(k^8)$ sets of integers $i_1 \ldots i_9$, the final running time is $O(n k^{12} \log^4 n)$.  

The third case is very similar to the second case.  We consider a case where $p^I$, the center on the interior of $H$, is above either the point on the left side or the right side of $H$.  If this is not true, we would perform the process symmetrically by guessing a center on the top side instead of the bottom side.  We can remove squares containing the first two center points the same way as in the second case.  When there are three center points remaining, we can still claim that one lies on the corner of $H$, but its not necessarily a bottom corner.  This just requires a few more cases to check.  It follows that this third case also takes $O(n k^{12} \log^4 n)$ time.

\subsection{The $(4,k)$- and $(5,k)$-Center Problem}
We can solve the original primal problem of determining whether a set of $p$ squares can contain all but $k$ points from an $n$ point set.  To find the minimum side length of the squares for this to be true we can use a matrix searching technique of Frederickson and Johnson \cite{FJ82,FJ83,FJ84} with $O(\log n)$ iterations of the above algorithm.  The minimal side length of a square is necessarily the difference in $x$-coordinates between two points or the difference in $y$-coordinates between two points.  We implicitly store these two orderings along the columns of two matrices, $X$ and $Y$, corresponding to the $x$- and the $y$-coordinates of the points.  The cells contain the differences in their values, but are only computed as needed.  Using monotone properties of these matrices we can search for the minimum such difference where our algorithm returns true.  We take the minimum from both matrices.  

\begin{theorem}
Given a set $P$ of $n$ points in $\b{R}^2$ and an integer $k \geq 0$, an optimal $(4,k)$-center of $P$ can be computed under the $\ell_\infty$-metric in $O(k^{O(1)} n \log n)$ time.
\end{theorem}

\begin{theorem}
Given a set $P$ of $n$ points in $\b{R}^2$ and an integer $k \geq 0$, an optimal $(5,k)$-center of $P$ can be computed under the $\ell_\infty$-metric in $O(k^{O(1)} n \log^5 n)$ time.  
\end{theorem}

\section*{Acknowledgements}
We thank Sariel Har-Peled for posing the problem and for several helpful discussions.  

\bibliographystyle{mystyle}
\bibliography{center-apx}

\end{document}